\documentclass[fontsize=11pt,a4paper,DIV=17]{scrartcl}

\usepackage{hyperref}
\usepackage{listings}
\usepackage{enumerate}

\usepackage{amsmath,amsfonts,amssymb,amsthm}
\usepackage{thmtools}

\usepackage{todonotes}
\usepackage{rsfso}
\usepackage{bibitemsep}

\usepackage[font=small,labelfont=bf]{caption}
\usepackage[font=small,labelfont=normalfont]{subcaption}

\newtheorem{theorem}{Theorem}
\newtheorem{lemma}[theorem]{Lemma}

\newtheorem{conjecture}{Conjecture}
\newtheorem{claim}{Claim}[theorem]
\newtheorem*{claim*}{Claim}

\newtheorem*{definition*}{Definition}

\def\AAsixA{{\cal N}_6^\Delta}  %%{\cal \hat{A}}_6}
\def\AA{{\cal A}}  
\def\BB{{\cal B}}  
\def\CC{{\cal C}}  

\def\qed{\hfill\fbox{\hbox{}}\medskip}
\def\qedclaim{\hfill$\triangle$\smallskip}

\input{psfig.sty}
\graphicspath{{fig/}}

\title{
Arrangements of Pseudocircles:\\ Triangles and Drawings\footnote{
Partially supported by DFG Grants FE 340/11-1 and FE 340/12-1.
Manfred Scheucher was partially supported by the ERC Advanced Research Grant no.~267165 (DISCONV).
The authors gratefully acknowledge the computing time granted 
by the Institute of Software Technology, Graz University of Technology.
}
}

\newcommand*\samethanks[1][\value{footnote}]{\footnotemark[#1]}

\author{Stefan Felsner\thanks{Institut f\"ur Mathematik, Technische Universit\"at Berlin, Germany, 
\texttt{\{felsner,scheucher\}@math.tu-berlin.de}}
\and Manfred Scheucher\samethanks}

\date{}
    
\begin{document}
\maketitle
\thispagestyle{plain}

\begin{abstract}
  A pseudocircle is a simple closed curve on the sphere or in the
  plane. The study of arrangements of pseudocircles was initiated by
  Gr\"unbaum, who defined them as collections of simple closed curves
  that pairwise intersect in exactly two crossings.  Gr\"unbaum
  conjectured that the number of triangular cells $p_3$ in digon-free
  arrangements of $n$ pairwise intersecting pseudocircles is at least
  $2n-4$. We present examples to disprove this conjecture. With a
  recursive construction based on an example with~$12$ pseudo\-circles
  and $16$ triangles we obtain a family of intersecting digon-free arrangements 
  with $p_3(\AA)/n \to 16/11 = 1.\overline{45}$.  
  We expect that the lower bound $p_3(\AA) \geq
  4n/3$ is tight for infinitely many simple arrangements.  It may
  however be true that all digon-free arrangements of~$n$ pairwise intersecting
  circles have at least $2n-4$ triangles.

  For pairwise intersecting arrangements with digons we have a lower bound of
  $p_3 \geq 2n/3$, and conjecture that $p_3 \geq n-1$.

  Concerning the maximum number of triangles in pairwise intersecting
  arrangements of pseudocircles, we show that $p_3 \le \frac{4}{3}\binom{n}{2} +O(n)$.  This is
  essentially best possible because there are families of pairwise
  intersecting arrangements of $n$ pseudocircles with $p_3 = \frac{4}{3}\binom{n}{2}$.

  The paper contains many drawings of arrangements of pseudocircles and 
  a good fraction of these drawings was produced automatically from the
  combinatorial data produced by our generation algorithm.  In the
  final section we describe some aspects of the drawing algorithm.
\end{abstract}

%%%%%%%%%%%%%%%%%%%%%%%%%%%%%%%%%%%%%%%%%%%%%%%%%%%%%%%%%%%%%%%%%%%%%%%%%%%%
\section{Introduction}

Arrangements of pseudocircles generalize arrangements of circles in the same
vein as arrangements of pseudolines generalize arrangements of lines. The
study of arrangements of pseudolines was initiated 1926 with an article of
Levi~\cite{l-dtdpe-26} where he proved the `Extension Lemma' and studied
triangles in arrangements. Since then arrangements of pseudolines were
intensively studied and the handbook article on the topic~\cite{fg-pa-16}
lists more than 100 references.

Gr\"unbaum~\cite{Gr72} initiated the study of arrangements of
pseudocircles. By stating a large number of conjectures he was
hoping to attract the attention of researchers for the topic. 
The success of this program was limited and several of Gr\"unbaum's 
45 year old conjectures remain unsettled. In this paper we
report on some progress regarding conjectures involving
numbers of triangles and digons in arrangements of
pseudocircles.

Some of our results and new conjectures are based on a
program written by the second author that enumerates all arrangements
of up to 7 pairwise intersecting pseudocircles. Before formally stating
our main results we introduce some terminology:
\goodbreak

An \emph{arrangement of pseudocircles} is a collection of closed
curves in the plane or on the sphere, called \emph{pseudocircles},
with the property that the intersection of any two of the
pseudocircles is either empty or consists of two points where the
curves cross.  An arrangement $\AA$ of pseudocircles is
 
\begin{description}
\item[\emph{simple},] if no three pseudocircles of $\AA$ intersect in a common point. 
\item[\emph{pairwise intersecting},]  if any two pseudocircles of $\AA$
  have non-empty intersection. We will frequently abbreviate and just write
  ``\emph{intersecting}'' instead of ``pairwise intersecting''.
\item[\emph{cylindrical},] if there are two cells of the
  arrangement which are separated by each of the pseudocircles.
\item[\emph{digon-free},] if there is no cell of the
  arrangement which is incident to only two pseudocircles.
\end{description}

We consider the sphere to be the most natural ambient space for
arrangements of pseudocircles. Consequently, we call two arrangements
isomorphic if they induce homeomorphic cell decompositions of the
sphere.  In many cases, in particular in all our figures, arrangements
of pseudocircles are embedded in the Euclidean plane, i.e., there is a
distinguished outer/unbounded cell. An advantage of such a
representation is that we can refer to the inner and outer side of a
pseudocircle.  Note that for every cylindrical arrangement of
pseudocircles it is possible to choose the unbounded cell such that
there is a point in the intersection of the inner discs of
all pseudocircles.

In an arrangement $\AA$ of pseudocircles, we denote a cell with $k$
crossings on its boundary as a \emph{$k$-cell} and let $p_k(\AA)$ be
the number of $k$-cells of~$\AA$.  Following Gr\"unbaum we call
2-cells \emph{digons} and remark that some other authors call them
\emph{lenses}. 3-cells are \emph{triangles}, 4-cells are
\emph{quadrangles}, and 5-cells are \emph{pentagons}.

In this paper we assume that arrangements of pseudocircles are simple unless
explicitly stated otherwise.

Conjecture~3.7 from Gr\"unbaum's monograph~\cite{Gr72} is:
{\em Every (not necessarily simple) digon-free arrangement of $n$
  pairwise intersecting pseudocircles has at least $2n-4$
  triangles}. Gr\"unbaum also provides examples of arrangements
of~$n\geq 6$ pseudocircles with $2n-4$ triangles. 

Snoeyink and Hersh\-berger~\cite{SH91} showed that the sweeping
technique, which serves as an important tool for the study of
arrangements of lines and pseudolines, can be adapted to 
work also in the case of arrangements of pseudocircles.
They used sweeps to show that, in an intersecting
arrangement of pseudocircles, every pseudocircle is incident to two
cells which are digons or triangles on either side. Therefore,
$2p_2 + 3p_3 \geq 4n$ which implies that every intersecting digon-free
arrangement of $n$ pseudocircles has at least $4n/3$ triangles.

Felsner and Kriegel~\cite{FK98} observed that the bound
from~\cite{SH91} also applies to non-simple intersecting digon-free
arrangements and gave examples of arrangements showing that the bound
is tight on this class for infinitely many values of~$n$.  These
examples disprove the conjecture in the non-simple case.

In Section~\ref{sec:few-tri}, we give counterexamples to Gr\"un\-baum's
conjecture \cite[Conjecture~3.7]{Gr72} which are simple.  
With a recursive construction based on an
example with~$12$ pseudo\-circles and $16$ triangles we obtain a family 
of digon-free intersecting arrangements with
$p_3/n \xrightarrow{n \to \infty} 16/11 = 1.\overline{45}$.  We then replace
Gr\"un\-baum's conjecture by Conjecture~\ref{conj:4/3_is_tight}: \emph{ The
  lower bound $p_3(\AA) \geq 4n/3$ is tight for infinitely many
  simple arrangements}.

A specific arrangement~$\AAsixA$ of~6 pseudocircles of $8$ triangles is
interesting in this context. The arrangement~$\AAsixA$ has no representation
with circles, two different proofs for the non-circularizablility of~$\AAsixA$
have been given in~\cite{fs-tap-17} and~\cite{FelsnerScheucher2019}.
The arrangement~$\AAsixA$ appears as a subarrangement in all known simple,
intersecting, digon-free arrangements with $p_3 < 2n-4$.  This motivates the
question, whether indeed Gr\"unbaum's conjecture is true when restricted to
intersecting arrangements of circles, see Conjecture~\ref{conj:weakGr}. In
Subsection~\ref{ssec:arrangement-with-digons} we discuss arrangements with
digons. We give an easy extension of the argument of Snoeyink and
Hersh\-berger~\cite{SH91} to show that these arrangements contain at least
$2n/3$ triangles. All intersecting arrangements known to us have at least
$n-1$ triangles and therefore our Conjecture~\ref{conj:general_lower} is that
$n-1$ is a tight lower bound for intersecting arrangements with digons.

In Section~\ref{sec:max-triangles} we study the maximum number of
triangles in arrangements of $n$ pseudocircles.  We show an upper
bound of order $2n^2/3+O(n)$. For the lower bound construction we glue two
arrangements of $n$ pseudolines into an arrangement of $n$
pseudocircles. 
Since respective arrangements of pseudolines are known,
we obtain arrangements of pseudocircles with $2n(n-1)/3$ triangles for
$n \equiv 0,4 \pmod 6$.

The paper contains many drawings of arrangements of pseudocircles and
a good fraction of these drawings was produced automatically from the
combinatorial data produced by the generation algorithm.
In Section~\ref{sec:tutte-draw} we describe some aspects of the
drawing algorithm which is based on iterative calls to
a Tutte embedding a.k.a.~spring embedding with adapting weights on the
edges.

%%%%%%%%%%%%%%%%%%%%%%%%%%%%%%%%%%%%%%%%%%%%%%%%%%%%%%%%%%%%%%%%%%%%%%%%%%%%
\section{Intersecting Arrangements with few Triangles}
\label{sec:few-tri}

The main result of this section is the following
theorem, which disproves Gr\"un\-{}baum's conjecture.

\begin{theorem}\label{thm:tri-main}
The minimum number of triangles in 
digon-free intersecting arrangements of $n$ pseudocircles is 
\begin{enumerate}[(i)]
\item \label{thm:tri-main:1}
$8$ for $3\leq n \leq 6$.
\item \label{thm:tri-main:2} 
$\lceil \frac{4}{3}n\rceil$ for $6 \leq n \leq 14$.
\item \label{thm:tri-main:3} 
$<\frac{16}{11}n$ for all $n =11k+1$ with $k\in \mathbb{N}$.
\end{enumerate}
\end{theorem}

Figures~\ref{fig:min_tri_3-4-5} and~\ref{fig:min_tri_6-7-8} show intersecting
arrangements with the minimum number of triangles for up to~$8$~pseudocircles.
We remark that, in total, there are three non-isomorphic intersecting
arrangements of $n=8$ pseudocircles with $p_3 = 11$ triangles, these are the
smallest counterexamples to Gr\"unbaum's conjecture
(cf. Lemma~\ref{lem:p3+2p2}).  We refer to our
website~\cite{scheucher_website} for further examples.

%%%%%%%%%%%%%%%%%%%%%%%%%%%%%%%%%%%%%%%%%%%%%%%%%%%%%%%%%%%%%%%%%%%%%%%%%%%%
\begin{figure}[htb]
  \centering
  \includegraphics[width=0.8\textwidth]{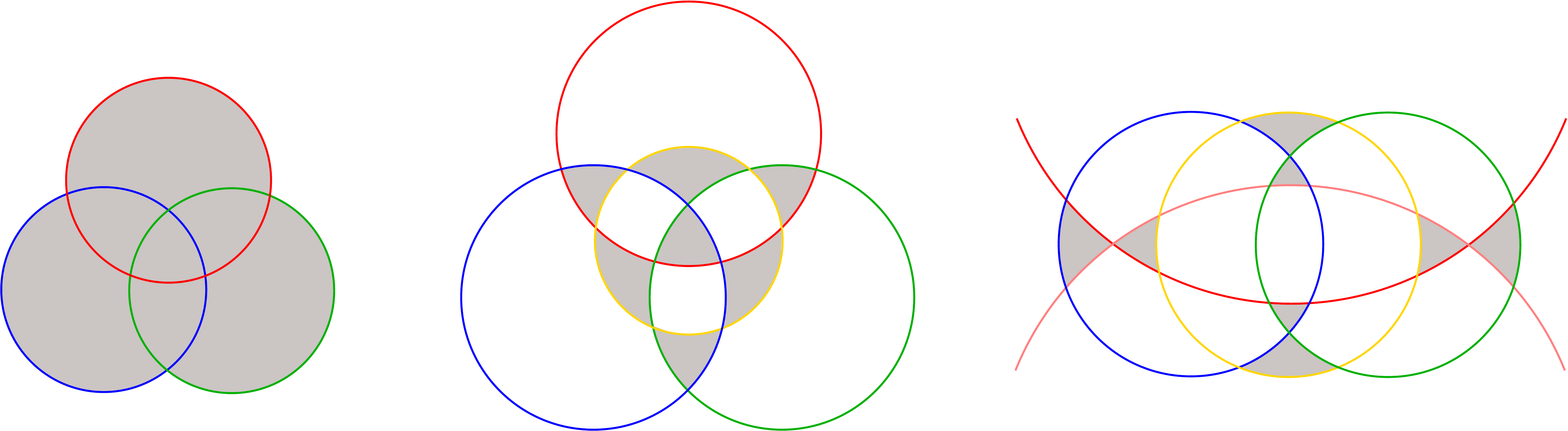}
  \caption{Digon-free intersecting arrangements of $n=3,4,5$ circles with $p_3=8$ triangles. Triangles (except the outer face) are colored gray.}
  \label{fig:min_tri_3-4-5}
\end{figure}
%%%%%%%%%%%%%%%%%%%%%%%%%%%%%%%%%%%%%%%%%%%%%%%%%%%%%%%%%%%%%%%%%%%%%%%%%%%%

The basis for Theorem~\ref{thm:tri-main} was laid by exhaustive computations,
which generated all intersecting arrangements of up to $n=7$ pseudocircles.
Starting with the unique intersecting arrangement of two pseudocircles, our
program recursively inserted pseudocircles in all possible ways.
From the complete enumeration, we know the minimum number of triangles for $n\leq
7$. In the range from 8 to 14, we had to iteratively use arrangements of $n$
pseudocircles with a small number of triangles and digons to generate
arrangements of $n+1$ pseudocircles with the same property.  Using this
strategy, we found intersecting arrangements with $\lceil 4n/3\rceil$
triangles for all $n$ in this range.  The corresponding lower bound
$p_3(\AA) \geq 4n/3$ is known from~\cite{SH91}.

The approach, which we had used to tackle arrangements of up to $n=14$
pseudocircles, made the complete enumerating of all arrangements obsolete.
However, since enumeration and counting is also much of interest in the
context of arrangements we decided to move the corresponding results
to~\cite{FelsnerScheucher2019}, where we investigate (not necessarily
intersecting) arrangements and focus on circularizability. The
arrangements and more information can be also be found on our companion
website~\cite{scheucher_website}.

%%%%%%%%%%%%%%%%%%%%%%%%%%%%%%%%%%%%%%%%%%%%%%%%%%%%%%%%%%%%%%%%%%%%%%%%%%%%
\begin{figure}[htb]
  \centering
   \hfill
  \includegraphics[width=0.29\textwidth]{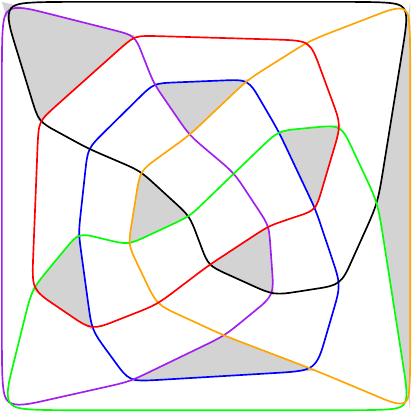}
   \hfill
  \includegraphics[width=0.29\textwidth]{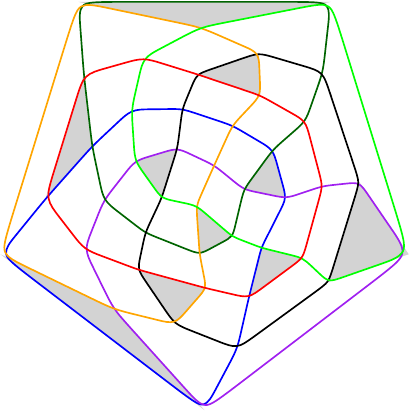}
   \hfill
  \includegraphics[width=0.29\textwidth]{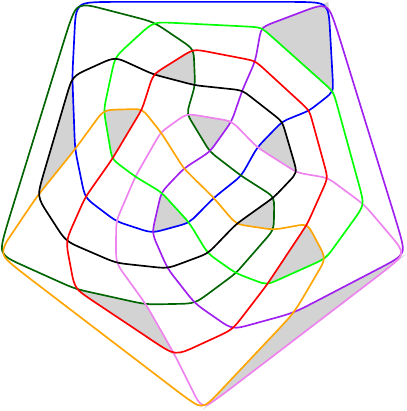}
   \hfill
  \caption{Digon-free intersecting arrangements of $n=6,7,8$ pseudocircles with $8$, $10$, $11$ triangles, respectively.
  The arrangement of $n=6$ pseudocircles on the left-hand side is named~$\AAsixA$.}
  \label{fig:min_tri_6-7-8}
\end{figure}
%%%%%%%%%%%%%%%%%%%%%%%%%%%%%%%%%%%%%%%%%%%%%%%%%%%%%%%%%%%%%%%%%%%%%%%%%%%%

Another result which we obtained from our computer search is the following:
the triangle-minimizing example for $n=6$ is unique, i.e., there is a unique
intersecting arrangement $\AAsixA$ of $6$ pseudocircles with $8$ triangles.
In~\cite{fs-tap-17} and~\cite{FelsnerScheucher2019} we gave 
different proofs for the non-circularizability of~$\AAsixA$.
Since the arrangement $\AAsixA$ appears as a subarrangement of all
arrangements with less than $2n-4$ triangles known to us,
the following weakening of Gr\"unbaum's conjecture might be true.

\begin{conjecture}[Weak Gr\"unbaum Conjecture]\label{conj:weakGr}
Every digon-free intersecting arrangement of $n$ circles has
at least $2n-4$ triangles.
\end{conjecture}

If this conjecture was true, it would imply a simple non-circularizability
criterion for intersecting arrangements: Any arrangement 
with $p_3 < 2n-4$ could directly be classified as
non-circularizable. 

So far we know that this conjecture is true for all $n\leq 9$.  The claim,
that we have checked all intersecting arrangements with $p_3(\AA) <
2n-4$ in this range, is justified by the following lemma, which restricts the
pairs $(p_2,p_3)$ for which there can exist arrangements of $n$ pseudocircles
whose extensions have $p_3(\AA) < 2n-4$. For example, to get all
digon-free intersecting arrangements of $n=9$ pseudocircles with $p_3 \le 13$ triangles, we
only had to extend intersecting arrangements of $n=7$ and $n=8$ pseudocircles 
with $p_3 + 2p_2 \leq 13$ triangles.

\begin{lemma}\label{lem:p3+2p2}
Let $\AA$ be an intersecting arrangement of pseudocircles.
Then for every subarrangement $\AA'$ of $\AA$ we have
$$
p_3(\AA')+ 2p_2(\AA') \le p_3(\AA) + 2p_2(\AA).
$$
\end{lemma} 

\begin{proof}
  We show the statement for a subarrangement $\AA'$ in which one
  pseudocircle~$C$ is removed from $\AA$. 
  The inequality then follows by iterating the argument.
  The arrangement $\AA'$ partitions the pseudocircle~$C$ into arcs.
  Reinsert these arcs one by one.

  Consider a triangle of~$\AA'$. 
  After adding an~arc, one of the following cases occurs:
  (1) the triangle remains untouched, or 
  (2) the triangle is split into a triangle and a quadrangle, or 
  (3) a digon is created in the region of the triangle. 
  
  Now consider a digon of~$\AA'$. 
  After adding an arc, one of the following cases occurs:
  (1)~the digon remains untouched, or 
  (2)~there is a new digon inside this digon, or 
  (3)~the digon has been split into two triangles.
\end{proof}

Levi~\cite{l-dtdpe-26} has shown that every arrangement of pseudolines in the
real projective plane has at least $n$ triangles.  Since arrangements of
great-(pseudo)circles are in bijection to arrangements of (pseudo)lines (the
bijection is explained in Section~\ref{sec:great-pseudocircles}), it directly
follows that every arrangement of great-pseudocircles has at least $2n$
triangles. The next theorem applies the same idea to 
a superclass of great-pseudocircle arrangements.
We think of the theorem as support of
the Weak Gr\"unbaum Conjecture (Conjecture~\ref{conj:weakGr}). 

\begin{theorem}
\label{thm:LeviGeneralized}
Let $\AA$ be an intersecting arrangement of $n$ pseudocircles such
that there is a pseudocircle $C$ in~$\AA$ that separates the two
intersection points $C' \cap C''$ of any other two pseudocircles $C'$
and~$C''$ in~$\AA$.  Then the number of triangles in $\AA$ is
at least~$2n$.
\end{theorem}

\begin{proof}
  Since, for every two pseudocircles $C'$ and~$C''$ distinct from~$C$, the two
  intersection points of $C' \cap C''$ are separated by the pseudocircle~$C$,
  the pseudocircle $C$ ``partitions'' the arrangement $\AA$ into two
  projective arrangements of $n$ pseudolines which lie in the two respective
  hemispheres.  According to Levi~\cite{l-dtdpe-26}, there are at least $n$
  triangles in each of the two arrangements, thus the original arrangement
  $\AA$ contains at least $2n$ triangles.
\end{proof}

Felsner and Kriegel~\cite{FK98} have shown that every arrangement of $n$
pseudolines in the Euclidean plane has at least $n-2$ triangles. 
This can again be turned into a result about triangles in
arrangements of pseudocircles.

\begin{theorem}
\label{thm:LeviGeneralizedPlus}
Let $\AA$ be an intersecting arrangement of $n$ pseudocircles.  If
$\AA$ can be extended by another pseudocircle $C$ such that the
pseudocircle $C$ separates the two intersection points $C' \cap C''$ of any
other two pseudocircles $C'$ and~$C''$, then the number of triangles in the
original arrangement $\AA$ is at least~$2n-4$.
\end{theorem}

\begin{proof}
  Since, for every two pseudocircles $C'$ and~$C''$ distinct from~$C$, the two
  intersection points of $C' \cap C''$ are separated by~$C$, the pseudocircle
  $C$ splits the arrangement $\AA$ into two Euclidean arrangements of
  $n$ pseudolines which lie in the two respective hemispheres.  According to
  Felsner and Kriegel~\cite{FK98}, there are at least $n-2$ triangles in each
  of the two arrangements.  Since the extending pseudocircle $C$ (which can be
  considered as the line at infinity in the respective Euclidean pseudoline
  arrangements) is not incident to any of these triangles, the arrangement
  $\AA$ contains at least $2n-4$ triangles.
\end{proof}

We now prepare for the proof of
Theorem~\ref{thm:tri-main}(\ref{thm:tri-main:3}), for which we construct a
family of (non-circularizable) intersecting arrangements of $n$ pseudocircles
with less than $\frac{16}{11}n$ triangles. The basis of
the construction is the arrangement $\AA_{12}$ of 12 pseudocircles
with 16 triangles shown in
Figure~\ref{fig:n12_p3_16_plus_recursion}(\subref{fig:n12_p3_16}). This
arrangement will be used iteratively for a `merge' as described by the
following lemma.

%%%%%%%%%%%%%%%%%%%%%%%%%%%%%%%%%%%%%%%%%%%%%%%%%%%%%%%%%%%%%%%%%%%%%%%%%%%%
\begin{figure}[htb]
  \centering
  \begin{subfigure}[t]{.48\textwidth}
    \includegraphics[page=1,width=0.95\textwidth]{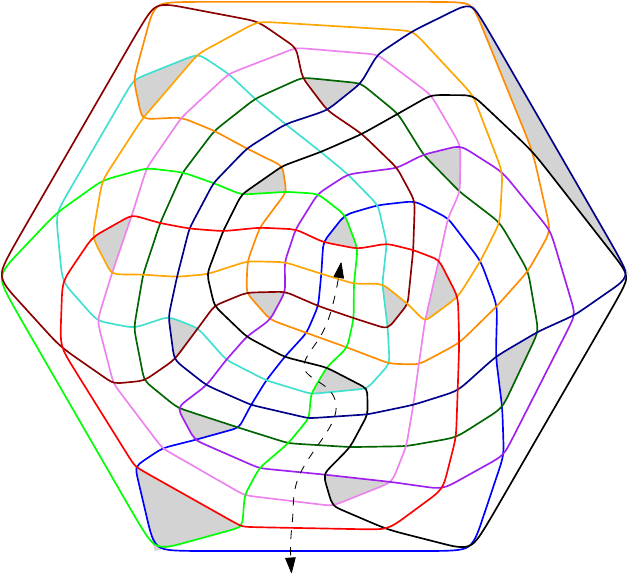}
    \caption{}
    \label{fig:n12_p3_16}
  \end{subfigure}
  \hfill
  \begin{subfigure}[t]{.5\textwidth}
    \includegraphics[page=1,width=0.95\textwidth]{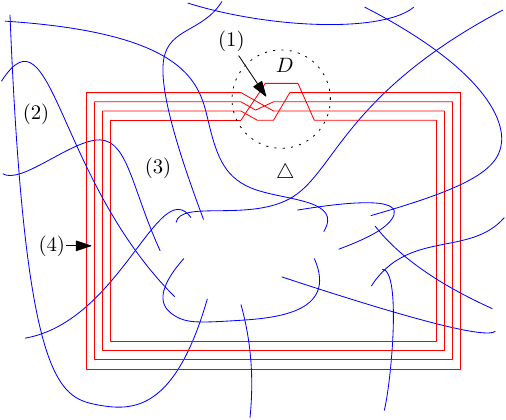}
    \caption{}
    \label{fig:recursion}  
  \end{subfigure}
  \caption{(\subref{fig:n12_p3_16})~The digon-free intersecting arrangement $\AA_{12}$ of 
    12 pseudocircles with exactly 16 triangles.
    The dashed curve intersects every pseudocircle exactly once.
    (\subref{fig:recursion})~An illustration of the construction in
      Lemma~\ref{lem:merge_A_plus_B}.
      Pseudocircles of~$\AA$ and $\BB$ are drawn red and blue, respectively.
  }
  \label{fig:n12_p3_16_plus_recursion} 
\end{figure}
%%%%%%%%%%%%%%%%%%%%%%%%%%%%%%%%%%%%%%%%%%%%%%%%%%%%%%%%%%%%%%%%%%%%%%%%%%%%

\begin{lemma}\label{lem:merge_A_plus_B}
  Let $\AA$ and $\BB$ be digon-free intersecting arrangements of
  $n_\AA \ge 3$ and~$n_\BB \ge 3$ pseudocircles, respectively.
  If there is a simple curve $P_\AA$ that 
  \begin{enumerate}[(1)]
   \item \label{item:lem:merge_A_plus_B:item1}
   intersects every pseudocircle of $\AA$ exactly once 
   \item 
   contains no vertex of $\AA$, 
   \item 
   traverses $\tau \ge 1$ triangles of~$\AA$, and 
   \item 
   forms $\delta$ triangles with pairs of pseudocircles from~$\AA$, 
  \end{enumerate}
  then there is a digon-free
  intersecting arrangement $\CC$ of $n_\AA+n_\BB-1$
  pseudocircles with $p_3(\CC) = p_3(\AA) +
  p_3(\BB) + \delta - \tau -1$ triangles.
\end{lemma}

We remark that condition~(\ref{item:lem:merge_A_plus_B:item1}) from the
statement of Lemma~\ref{lem:merge_A_plus_B} asserts that $\AA$ is
cylindrical.  Moreover, if $\BB$ is cylindrical, then also
$\CC$ is cylindrical.

\begin{proof}
  Take a drawing of~$\AA$ and make a hole in the two cells
  which contain the ends of~$P_\AA$.  This yields a drawing
  of~$\AA$ on a cylinder such that none of the pseudocircles
  is contractible.  The path $P_\AA$
  connects the two boundaries of the cylinder. In fact, the existence
  of a path with the properties of $P_\AA$ characterizes
  cylindrical arrangements. 

  Stretch the cylindrical 
  drawing such that it becomes a narrow belt, where all intersections
  of pseudocircles take place in a small disk, which we call
  \emph{belt-buckle}. This drawing of~$\AA$ is called a
  \emph{belt drawing}. The drawing of the red subarrangement in
  Figure~\ref{fig:n12_p3_16_plus_recursion}(\subref{fig:recursion}) shows a belt drawing.

Choose a triangle $\triangle$ in $\BB$ and a pseudocircle $B$ which is
incident to $\triangle$. Let~$b$ be the \emph{edge} of~$B$ on the
boundary of~$\triangle$. 
Specify a disk~$D$, which is traversed by~$b$
and disjoint from all other edges of~$\BB$. Now
replace~$B$ by a belt drawing of~$\AA$ in a small neighborhood of~$B$
such that the belt-buckle is drawn within~$D$; 
see Figure~\ref{fig:n12_p3_16_plus_recursion}(\subref{fig:recursion}).

The arrangement $\CC$ obtained from 
\emph{merging} $\AA$ and $\BB$, as we just described, 
has $n_\AA + n_\BB - 1$ pseudocircles. Moreover
if $\AA$ and $\BB$ are digon-free/intersecting, then
$\CC$ has the same property.
Most of the cells $c$ of~$\CC$ are of one of the following four types: 
\begin{enumerate}[(1)]
\item All boundary edges of $c$ belong to pseudocircles of~$\AA$.  
\item All boundary edges of $c$ belong to pseudocircles of~$\BB$.  
\item All but one of the boundary edges of $c$ belong to pseudocircles 
of~$\BB$ and the remaining edge belongs to $\AA$.
(These cells correspond to cells of~$\BB$ with a boundary edge on~$B$.) 
\item Quadrangular cells, whose boundary edges alternatingly
belong to $\AA$ and~$\BB$.
\end{enumerate}

From the cells of~$\BB$, only $\triangle$ and the other cell
containing $b$ (which is not a digon since $\BB$ is
digon-free) have not been taken into account.  In~$\CC$, the
corresponding two cells have at least two boundary edges from
$\BB$ and at least two from~$\AA$. Consequently,
neither of the two cells are triangles.  The remaining cells
of~$\CC$ are bounded by pseudocircles from $\AA$
together with one of the two bounding pseudocircles of~$\triangle$
other than~$B$. These two pseudocircles cross through~$\AA$
following the path prescribed by $P_\AA$.  There are~$\delta$
triangles among these cells, but~$\tau$ of these are obtained
because~$P_\AA$ traverses a triangle of~$\AA$. Among
cells of~$\CC$ of types (1) to (4) all the triangles have a
corresponding triangle in~$\AA$ or~$\BB$.
But~$\triangle$ is a triangle of $\BB$ which does not occur in
this correspondence.  Hence, there are $p_3(\AA) +
p_3(\BB) +\delta-\tau-1$ triangles in~$\CC$.  
\end{proof}

%\begin{proof}[Proof of Theorem~\ref{thm:tri-main}(\ref{thm:tri-main:3})]
\noindent
\emph{Proof of Theorem~\ref{thm:tri-main}(\ref{thm:tri-main:3}).}
  We use~$\AA_{12}$, the arrangement shown 
  in Figure~\ref{fig:n12_p3_16_plus_recursion}(\subref{fig:n12_p3_16}),
  in the role of~$\AA$ for our recursive construction.  The dashed
  path in the figure is used as~$P_\AA$ with $\delta=2$ and $\tau=1$.
  Starting with $\CC_1=\AA_{12}$ and defining $\CC_{k+1}$ as the merge of
  $\CC_k$ and~$\AA_{12}$, we construct a sequence $\{\CC_{k}\}_{k \in
    \mathbb{N}}$ of digon-free intersecting arrangements of $n(\CC_{k}) = 11k+1$
  pseudocircles with $p_3(\CC_{k}) = 16k$ triangles.  The fraction
  $16k/(11k+1)$ is increasing with~$k$ and converges to $16/11 =
  1.\overline{45}$ as~$n$ goes to $\infty$.
  \qed
\medskip

We remark that using other arrangements from
Theorem~\ref{thm:tri-main}(\ref{thm:tri-main:2}) (which also admit a
path with $\delta=2$ and $\tau=1$) in the recursion, we obtain 
intersecting arrangements with $p_3 = \lceil \frac{16}{11}n \rceil$ triangles for
all $n \geq 6$.

Since the lower bound $\lceil\frac{4}{3}n\rceil$ is tight for $6 \le n
\le 14$, we believe that the following is true:

\begin{conjecture}\label{conj:4/3_is_tight}
There are digon-free intersecting arrangements $\AA$ of $n$ pseudocircles with 
$p_3(\AA) = \lceil 4n/3 \rceil$ for infinitely many values of~$n$.
\end{conjecture}

%%%%%%%%%%%%%%%%%%%%%%%%%%%%%%%%%%%%%%%%%%%%%%%%%%%%%%%%%%%%%%%%%%%%%%%%%%%%
\subsection{Intersecting Arrangements with Digons}
\label{ssec:arrangement-with-digons}

We know intersecting arrangements of $n \ge 3$ pseudocircles with digons 
and only $n-1$ triangles. The examples depicted
in Figure~\ref{fig:construction_n_minus_1_triangles_with_digons}
are part of an infinite family of such arrangements.
As illustrated, the intersection order with the black circle determines the arrangement.
In fact, it is easy to see that $2^{n-3}$ different arrangements are possible:
Starting with the black, the purple, and the yellow pseudocircles (which give a unique arrangement),
each further pseudocircle has its finger placed 
either immediately to the left or immediately to the right of the previous finger.
Figures~\ref{fig:construction_n_minus_1_triangles_with_digons}(\subref{fig:construction_n_minus_1_triangles_with_digons:A}) 
and~\ref{fig:construction_n_minus_1_triangles_with_digons}(\subref{fig:construction_n_minus_1_triangles_with_digons:B}) 
illustrate the finger-insertion-sequence ``right-right-right-\ldots'' and ``left-right-left-\ldots'', respectively.

%%%%%%%%%%%%%%%%%%%%%%%%%%%%%%%%%%%%%%%%%%%%%%%%%%%%%%%%%%%%%%%%%%%%%%%%%%%%
\begin{figure}[htb]
  \centering
  
  \hbox{}
  \hfill
  \begin{subfigure}[t]{.43\textwidth}
    \centering
    \includegraphics[width=0.9\textwidth]{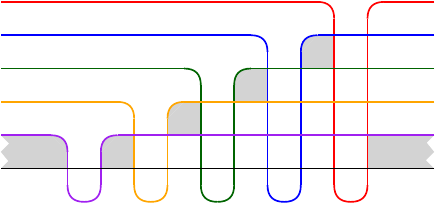}
    \caption{}
    \label{fig:construction_n_minus_1_triangles_with_digons:A}
  \end{subfigure}
  \hfill
  \begin{subfigure}[t]{.43\textwidth}
    \centering
    \includegraphics[width=0.9\textwidth,page=2]{few-tri-with-digons_ipe_B_curved}
    \caption{}
    \label{fig:construction_n_minus_1_triangles_with_digons:B}
  \end{subfigure}
  \hfill
  \hbox{}
  
  \caption{Intersecting arrangements of $n$~pseudocircles with $n$ digons and $n-1$ triangles.
    \label{fig:construction_n_minus_1_triangles_with_digons}
  }
\end{figure}
%%%%%%%%%%%%%%%%%%%%%%%%%%%%%%%%%%%%%%%%%%%%%%%%%%%%%%%%%%%%%%%%%%%%%%%%%%%%

Using ideas based on sweeps (cf.~\cite{SH91}), 
we can show that every pseudocircle 
is incident to at least two triangles.
This implies the following theorem:

\begin{theorem}\label{thm:2n_3_triangles}
Every intersecting arrangement of $n \ge 3$ pseudocircles has
at least $2n/3$ triangles.
\end{theorem}

The proof of the theorem is based on the following lemma:

\begin{lemma}\label{lem:in-out}
  Let $C$ be a pseudocircle in an intersecting arrangement of $n \ge 3$
  pseudocircles.  Then all digons incident to $C$ lie on the same side
  of~$C$.
\end{lemma}
\begin{proof}
  Consider a pseudocircle $C'$ that forms a digon $D'$ with $C$ that
  lies, say, ``inside'' $C$.  If $C''$ also forms a digon $D''$ with~$C$, then
  $C''$ has to cross $C'$ in the exterior of $C$. Hence $D''$ also has to lie
  ``inside''~$C$.  Consequently, all digons incident to $C$ lie on the
  same side of~$C$.
\end{proof}

%\begin{proof}[Proof of Theorem~\ref{thm:2n_3_triangles}]
\noindent
\emph{Proof of Theorem~\ref{thm:2n_3_triangles}.}
  Let $\AA$ be an intersecting arrangement and consider a drawing of
  $\AA$ in the plane. Snoeyink and Hershberger~\cite{SH91} have shown that
  starting with any circle~$C$ from $\AA$ the outside of $C$ can be swept
  with a closed curve $\gamma$ until all of the arrangement is inside
  of $\gamma$. During the sweep $\gamma$ intersects every
  pseudocircle from $\AA$ at most twice. The sweep uses two
  types\footnote{%
    There is a third type of move \emph{take a hump} 
    which is the inverse of ``leave a pseudocircle''.
    However, this third type does not occur in the proof of Theorem~\ref{thm:2n_3_triangles} 
    because each two pseudocircles already intersect.} 
  of moves to make progress: 
\begin{enumerate}[(1)]
\item 
\emph{take a crossing}, in~\cite{SH91} this is called `pass a triangle'; 
\item 
\emph{leave a pseudocircle}, this is possible when $\gamma$ and
some pseudocircle form a digon which is on the outside of $\gamma$,
in~\cite{SH91} this is called `pass a hump'.
\end{enumerate}
Figure~\ref{fig:the2flips} gives an illustration of the two possible types of moves.

%%%%%%%%%%%%%%%%%%%%%%%%%%%%%%%%%%%%%%%%%%%%%%%%%%%
\begin{figure}[htb]
  \centering
  \includegraphics{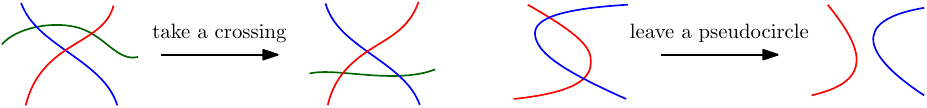}
  \caption{An illustration of the two types of moves which are possible in the proof of Theorem~\ref{thm:2n_3_triangles}.
The blue curve is~$\gamma$. The interior of $\gamma$ is left of the shown part of the curve.}
  \label{fig:the2flips}
\end{figure}
%%%%%%%%%%%%%%%%%%%%%%%%%%%%%%%%%%%%%%%%%%%%%%%%%%%

Let $C$ be a pseudocircle of $\AA$.  By the previous lemma, all digons
incident to $C$ lie on the same side of~$C$.  Redraw $\AA$ so that all
digons incident to $C$ are inside $C$.  The first move of a sweep starting at
$C$ has to take a crossing, and hence, there is a triangle $\triangle$
incident to $C$.  Redraw $\AA$ such that $\triangle$ becomes the
unbounded face.  Again consider a sweep starting at $C$. The first move of
this sweep reveals a triangle~$\triangle'$ incident to $C$. Since $\triangle$
is not a bounded triangle of the new drawing we have $\triangle \neq
\triangle'$, and hence, $C$ is incident to at least two triangles.  The proof
is completed by double counting the number of incidences of triangles and
pseudocircles.  \qed \medskip

Since for $3 \le n \le 7$ every intersecting arrangement
has at least $n-1$ triangles, we believe that the following is true:

\begin{conjecture}\label{conj:general_lower}
Every intersecting arrangement of $n \ge 3$ pseudocircles has
at least $n-1$ triangles.
\end{conjecture}

If the arrangement is not required to be intersecting, then the proof of
Lemma~\ref{lem:in-out} fails. Indeed, if the intersection graph of the
arrangement is bipartite, then all faces are of even degree, in particular,
there are no triangles; see
Figure~\ref{fig:interesting_connected_arrangements}(\subref{fig:construction_0_triangles_connected}).

%%%%%%%%%%%%%%%%%%%%%%%%%%%%%%%%%%%%%%%%%%%%%%%%%%%%%%%%%%%%%%%%%%%%%%%%%%%%
\begin{figure}[htb]
  \centering
  
  \hbox{}
  \hfill
  \begin{subfigure}[t]{.45\textwidth}
    \centering
    \includegraphics{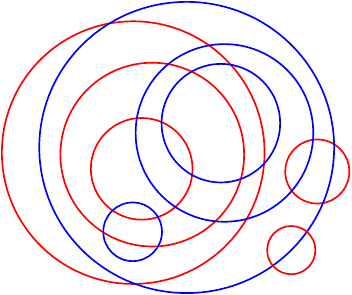}
    \caption{}
    \label{fig:construction_0_triangles_connected}
  \end{subfigure}
  \hfill
  \begin{subfigure}[t]{.45\textwidth}
  \centering
    \includegraphics{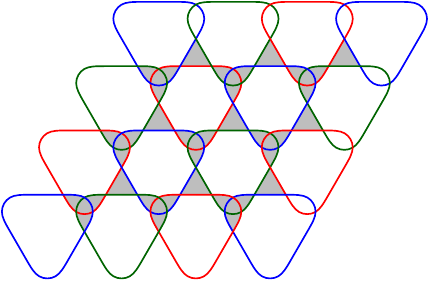}
    \caption{}
  \label{fig:connected_arrangement_56_triangles}
  \end{subfigure}
  \hfill
  \hbox{}
  
  \vspace{0.5cm}
  \begin{subfigure}[t]{.65\textwidth}
  \centering
    \includegraphics{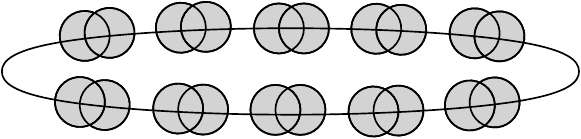}
    \caption{}
  \label{fig:connected_arrangement_all_triangles}
  \end{subfigure}
  
  \caption{Non-intersecting arrangements
    (\subref{fig:construction_0_triangles_connected})~with no triangles,
    (\subref{fig:connected_arrangement_56_triangles})~with a triangle-cell-ratio of $5/6+O(1/\sqrt{n})$, and
    (\subref{fig:connected_arrangement_all_triangles})~with only two non-triangular cells,
    i.e., with a triangle-cell-ratio of~$1+O(1/n)$.
  }
  \label{fig:interesting_connected_arrangements}
\end{figure}
%%%%%%%%%%%%%%%%%%%%%%%%%%%%%%%%%%%%%%%%%%%%%%%%%%%%%%%%%%%%%%%%%%%%%%%%%%%%

%%%%%%%%%%%%%%%%%%%%%%%%%%%%%%%%%%%%%%%%%%%%%%%%%%%%%%%%%%%%%%%%%%%%%%%%%%%%
\section{Maximum Number of Triangles}
\label{sec:max-triangles}

Regarding the maximum number of triangles the complete enumeration\footnote{
While we only have a complete database of
arrangements of up to 7~pseudocircles (cf.\ Table~1 in
\cite{FelsnerScheucher2019}) we could make sure that there is no
arrangement of 8 pseudocircles with at least 39 triangles.  Such an arrangement
$\AA$ would have a pseudocircle $C$ such that the number of triangles
of $\AA' = \AA - C$ would be at least 25.  Our computations showed
that no arrangement $\AA'$ of $7$ pseudocircles with $p_3(\AA') \ge
25$ can be extended to an arrangement of 8 pseudocircles with more than 38
triangles.
}
provides precise data for $n\leq 8$.
Moreover, we used heuristics to generate examples with many triangles for larger~$n$.  
Table~\ref{table:h3_table_upper} summarizes our results and
Figures~\ref{fig:n5_n6_triangle_maximal} and~\ref{fig:n7_n8_triangle_maximal} show
intersecting arrangements of $n=5,6,7,8$ pseudocircles 
with the maximal number of triangles; 
further arrangements are available on our website~\cite{scheucher_website}.

%%%%%%%%%%%%%%%%%%%%%%%%%%%%%%%%%%%%%%%%%%%%%%%%%%%%%%%%%%%%%%%%%%%%%%%%%%%%
\begin{figure}[htb]
  \centering
  \begin{subfigure}[t]{.22\textwidth}
    \centering
    \includegraphics[width=\textwidth]{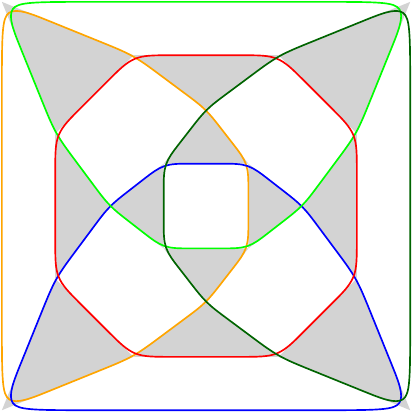}
    \caption{}
    \label{fig:n5_p3_12_digonfree}
  \end{subfigure}
  \hfill
  \begin{subfigure}[t]{.25\textwidth}
    \centering
    \includegraphics[width=\textwidth]{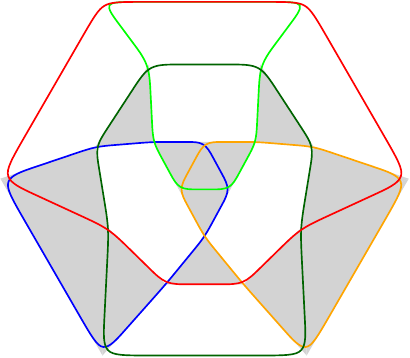}
    \caption{}
    \label{fig:n5_p3_13_digons}
  \end{subfigure}
  \hfill
  \begin{subfigure}[t]{.23\textwidth}
    \centering
    \includegraphics[width=\textwidth]{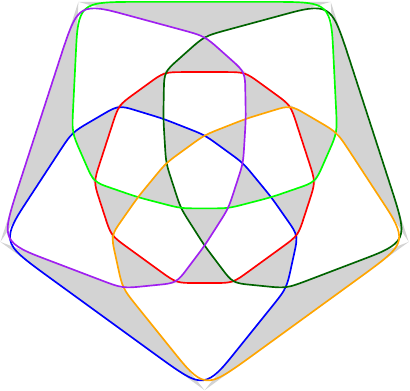}
    \caption{}
    \label{fig:n6_20triangles_df_icosidodecahedron}
  \end{subfigure}
  \hfill
  \begin{subfigure}[t]{.23\textwidth}
    \centering
    \includegraphics[width=\textwidth]{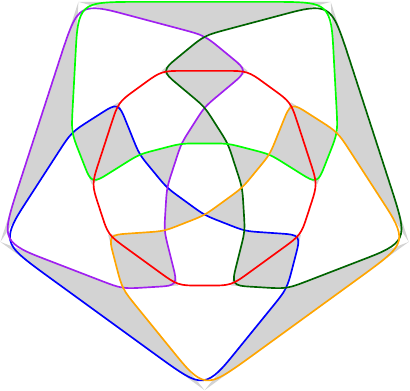}
    \caption{}
    \label{fig:n6_20triangles_df_shifted_icosidodecahedron}
  \end{subfigure}
  
  \caption{
  (\subref{fig:n5_p3_12_digonfree})~and (\subref{fig:n5_p3_13_digons})~show arrangements of $n=5$ pseudocircles. 
  The first one is digon-free and has 12 triangles and the second one has 13 triangles and one digon. 
  (\subref{fig:n6_20triangles_df_icosidodecahedron})~and 
  (\subref{fig:n6_20triangles_df_shifted_icosidodecahedron})~show arrangements of $n=6$ with 20 triangles. 
  The arrangement in (\subref{fig:n6_20triangles_df_icosidodecahedron}) is the skeleton of the Icosidodecahedron.
  }
  \label{fig:n5_n6_triangle_maximal}
\end{figure}	
%%%%%%%%%%%%%%%%%%%%%%%%%%%%%%%%%%%%%%%%%%%%%%%%%%%%%%%%%%%%%%%%%%%%%%%%%%%%

%%%%%%%%%%%%%%%%%%%%%%%%%%%%%%%%%%%%%%%%%%%%%%%%%%%%%%%%%%%%%%%%%%%%%%%%%%%%
\begin{figure}[htb]
  \centering
  
  \hbox{}
  \hfill
  \begin{subfigure}[t]{.32\textwidth}
    \centering
    \includegraphics[width=\textwidth]{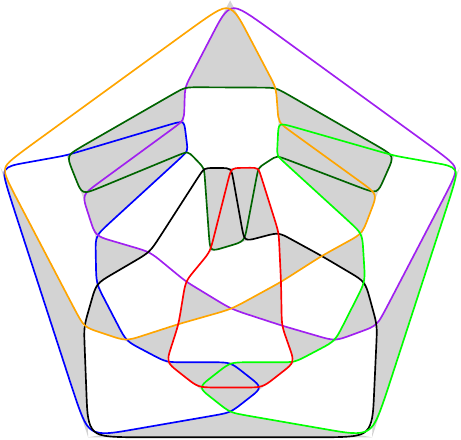}
    \caption{}
    \label{fig:n7_p3_29_digonfree}
  \end{subfigure}
  \hfill
  \begin{subfigure}[t]{.32\textwidth}
    \centering
    \includegraphics[width=\textwidth]{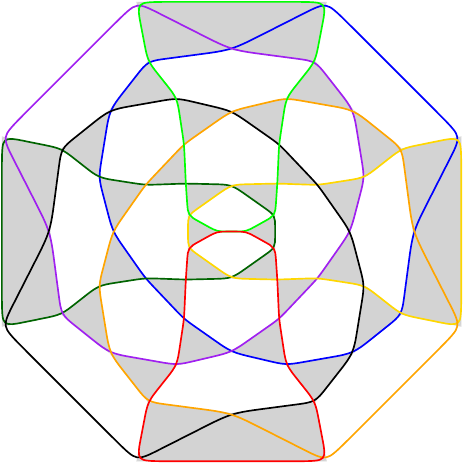}
    \caption{}
    \label{fig:n8_p3_38_digonfree}
  \end{subfigure}
  \hfill
  \hbox{}
  
  \caption{
  (\subref{fig:n7_p3_29_digonfree})~An arrangement of $n=7$ pseudocircles with 29 triangles.\\
  (\subref{fig:n8_p3_38_digonfree})~An arrangement of $n=8$ pseudocircles with 38 triangles.
  }
  \label{fig:n7_n8_triangle_maximal}
\end{figure}	
%%%%%%%%%%%%%%%%%%%%%%%%%%%%%%%%%%%%%%%%%%%%%%%%%%%%%%%%%%%%%%%%%%%%%%%%%%%%

In the next subsection we show that asymptotically the contribution of edges
that are incident to two triangles is neglectable.  The last subsection gives
a construction of intersecting arrangements which show that
$\lfloor\tfrac{4}{3}\binom{n}{2}\rfloor$ is attained for infinitely many
values of $n$.

%%%%%%%%%%%%%%%%%%%%%%%%%%%%%%%%%%%%%%%%%%%%%%%%%%%%%%%%%%%%%%%%%%%%%%%%%%%%
\begin{table}[htb]\centering
\advance\tabcolsep5pt
\def\arraystretch{1.2}
\begin{tabular}{ l|r|r|r|r|r|r|r|r|r}
$n$		&2 	&3 	&4 	&5 	&6 	&7	&8	&9	&10\\
\hline
simple		&0	&8	&8	&13	&20	&29	&38	&$\ge 48$	&$\ge 60$\\
digon-free	&-	&8	&8	&12	&20	&29	&38	&$\ge 48$	&$\ge 60$\\
$\lfloor\tfrac{4}{3}\binom{n}{2}\rfloor$	
		&1	&4	&8	&13	&20	&28	&37	&48	&60
\end{tabular}
%\vskip4mm
\caption{Maximum number of triangles in intersecting arrangements of $n$ pseudocircles.}
\label{table:h3_table_upper}
% \vskip-4mm
\end{table}
%%%%%%%%%%%%%%%%%%%%%%%%%%%%%%%%%%%%%%%%%%%%%%%%%%%%%%%%%%%%%%%%%%%%%%%%%%%%

Recall that we only study simple intersecting
arrangements. Gr\"unbaum~\cite{Gr72} also looked at non-simple
arrangements. His Figures 3.30, 3.31, and 3.32 show 
drawings of simplicial arrangements that have $n=7$ with $p_3= 32$,
$n=8$ with $p_3= 50$, and $n=9$ with $p_3= 62$, respectively.
Hence, non-simple arrangements can have more triangles.

\begin{theorem}\label{thm:upper-max}
Every intersecting arrangement $\AA$ of pseudocircles fulfills
$p_3(\AA)  \le \frac{4}{3}\binom{n}{2}+O(n)$.
\end{theorem}

\begin{proof}
 
  Let $\AA$ be an intersecting arrangement of $n\ge 4$ pseudocircles.  We
  view $\AA$ as a 4-regular plane graph, i.e., the set $X$ of
  crossings is the vertex set and edges are the segments which connect
  consecutive crossings on a pseudocircle.

\begin{claim}
\label{thm:upper-max_claim:A}
No crossing is incident to 4 triangular cells.
\end{claim}
Assume that a crossing $u$ of $C_i$ and $C_j$ is incident to
four triangular cells. Then there is a pseudocircle~$C_k$ which bounds those 4
triangles, see 
Figure~\ref{fig:proof_maxtriangles_14}(\subref{fig:proof_maxtriangles_1}). 
Now $C_k$ only intersects $C_i$ and
$C_j$. This, however, is impossible because $n \ge 4$ and
$\AA$ is intersecting.\qedclaim

%%%%%%%%%%%%%%%%%%%%%%%%%%%%%%%%%%%%%%%%%%%%%%%%%%%%%%%%%%%%%%%%%%%%%%%%%%%%
\begin{figure}[htb]
  \centering
  
  \hbox{}
  \hfill
  \begin{subfigure}[t]{.3\textwidth}
    \centering
    \includegraphics[page=1]{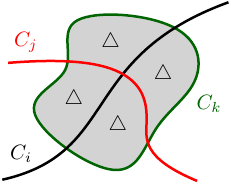}
    \caption{}
    \label{fig:proof_maxtriangles_1}
  \end{subfigure}
  \hfill
  \begin{subfigure}[t]{.4\textwidth}
    \centering
    \includegraphics[page=4]{proof_maxtriangles_revised}
    \caption{}
    \label{fig:proof_maxtriangles_4}
  \end{subfigure}
  \hfill
  \hbox{}
 
  \caption{Illustrations of the proof of 
Claim~\ref{thm:upper-max_claim:A} and Claim~\ref{thm:upper-max_claim:B}.}
  \label{fig:proof_maxtriangles_14}
\end{figure}
%%%%%%%%%%%%%%%%%%%%%%%%%%%%%%%%%%%%%%%%%%%%%%%%%%%%%%%%%%%%%%%%%%%%%%%%%%%%

Let $X' \subseteq X$ be the set of crossings of $\AA$
that are incident to 3 triangular cells. Our aim is to show that
$|X'|$ is small, in fact $|X'| \in O(n)$.
When this is shown we can bound the number of triangles in
$\AA$ as follows:
Under the assumption that $|X'| \in O(n)$,
the number of triangles incident to a crossing in $X'$ 
clearly is in $O(n)$. Now let $Y = X\setminus X'$.
Each of the remaining triangles is incident to three elements of
$Y$ and each crossing of $Y$ is incident to at most 2 triangles.
Hence, there are at most $2|Y|/3 + O(n)$ triangles.
Since $|Y| \leq |X| = n(n-1)$ we obtain the bound claimed in the
statement of the theorem.
 
To show that $|X'|$ is small we need some preparation.
 
\begin{claim}
\label{thm:upper-max_claim:B}
Two adjacent crossings $u,v$ in $X'$ share two triangles.
\end{claim}
Since $u$ and $v$ are both incident to 3 triangles, there is at least
one triangle $\triangle$ incident to both of them.  Assume for a
contradiction that the other cell which is incident to the segment
$uv$ is not a triangle.  Let $C_i,C_j,C_k$ be the three pseudocircles
such that $u$ is a crossing of $C_i$ and $C_j$, $v$ is a
crossing of $C_i$ and $C_k$, and $\triangle$ is bounded by
$C_i,C_j,C_k$; see Figure~\ref{fig:proof_maxtriangles_14}(\subref{fig:proof_maxtriangles_4}).  
We denote
the third vertex of $\triangle$ by~$w$ and note that $w$ is a crossing
of $C_j$ and $C_k$.

Since $u$ is incident to three triangles,
the segment $uw$ bounds another triangle,
which is again defined by $C_i,C_j,C_k$.
Let $u'$ be the third vertex incident to that triangle.
Similarly, the segment $vw$ is incident to another triangle 
which is also defined by $C_i,C_j,C_k$, 
and has a third vertex $v'$.

Again, by the same argument, the segments $uu'$ and $vv'$, respectively, are
both incident to another triangle.  However, this is impossible as the two
circles $C_j$ and $C_k$ intersect three times.  Thus both faces incident to
segment $uv$ are triangles.  \qedclaim

\begin{claim}
\label{thm:upper-max_claim:C}
Let $u,v,w$ be three distinct crossings in $X'$.
If $u$ is adjacent to both $v$ and $w$, then $v$ is adjacent to $w$.
\end{claim}
Since $u$ is incident to three triangles and
the segments $uv$ and $uw$ are both incident to two triangles,
there is a triangle $\triangle$ with corners $u,v,w$.
This triangle shows that $u$, $v$, and $w$ are adjacent to each other.
\qedclaim

Claim~\ref{thm:upper-max_claim:C}
implies that each connected component of the graph induced by $X'$ is
a complete graph. It is easy to see that a $K_4$ induced by $X'$ 
is impossible, and therefore, all components induced by $X'$ are either singletons,
edges, or triangles. Figure~\ref{fig:proof_maxtriangles_567} shows
the local structure of the arrangement around components of these
three types. 

%%%%%%%%%%%%%%%%%%%%%%%%%%%%%%%%%%%%%%%%%%%%%%%%%%%%%%%%%%%%%%%%%%%%%%%%%%%%
\begin{figure}[htb]
  \centering
  \begin{subfigure}[t]{.27\textwidth}
    \centering
    \includegraphics[width=\textwidth,page=5]{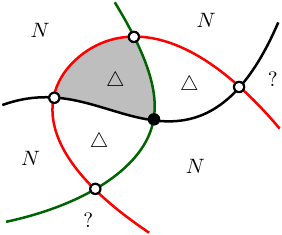}
    \caption{}
    \label{fig:proof_maxtriangles_5}
  \end{subfigure}
  \hfill
  \begin{subfigure}[t]{.27\textwidth}
    \centering
    \includegraphics[width=\textwidth,page=6]{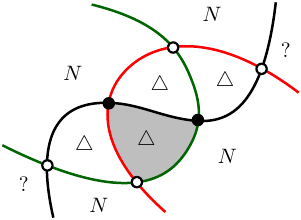}
    \caption{}
    \label{fig:proof_maxtriangles_6}
  \end{subfigure}
  \hfill
  \begin{subfigure}[t]{.27\textwidth}
    \centering
    \includegraphics[width=\textwidth,page=7]{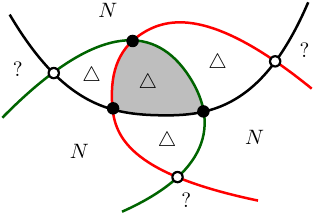}
    \caption{}
    \label{fig:proof_maxtriangles_7}
  \end{subfigure}
 
  \caption{An illustration of the configurations of crossings in $X'$. 
  In this figure $\triangle$ marks a triangle,
  ``$N$'' marks a $k$-cell with $k \ge 4$ 
  (``neither a triangle, nor a digon''), 
  ``?'' marks an arbitrary cell.
  Crossings with 3 incident triangles are shown as black vertices
  (these are the crossings in $X'$).
}
  \label{fig:proof_maxtriangles_567}
\end{figure}
%%%%%%%%%%%%%%%%%%%%%%%%%%%%%%%%%%%%%%%%%%%%%%%%%%%%%%%%%%%%%%%%%%%%%%%%%%%%

%%%%%%%%%%%%%%%%%%%%%%%%%%%%%%%%%%%%%%%%%%%%%%%%%%%%%%%%%%%%%%%%%%%%%%%%%%%%
\begin{figure}[htb]
  \centering
    \begin{subfigure}[t]{.27\textwidth}
    \centering
    \includegraphics[width=\textwidth,page=9]{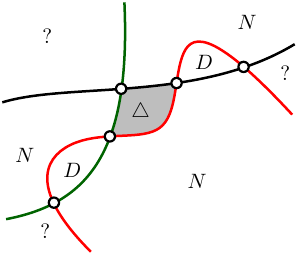}
    \caption{}
    \label{fig:proof_maxtriangles_9}
  \end{subfigure}
  \hfill
  \begin{subfigure}[t]{.27\textwidth}
    \centering
    \includegraphics[width=\textwidth,page=10]{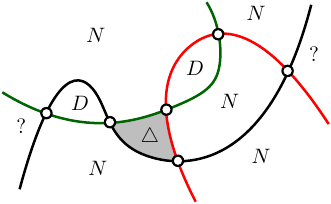}
    \caption{}
    \label{fig:proof_maxtriangles_10}
  \end{subfigure}
  \hfill
  \begin{subfigure}[t]{.27\textwidth}
    \centering
    \includegraphics[width=\textwidth,page=11]{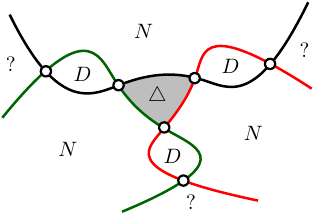}
    \caption{}
    \label{fig:proof_maxtriangles_11}
  \end{subfigure}
\caption{The configurations in (\subref{fig:proof_maxtriangles_9}), (\subref{fig:proof_maxtriangles_10}), and (\subref{fig:proof_maxtriangles_11}) 
are obtained by flipping the gray triangle in the configuration from
Figure~\ref{fig:proof_maxtriangles_567}(\subref{fig:proof_maxtriangles_5}),
\ref{fig:proof_maxtriangles_567}(\subref{fig:proof_maxtriangles_6}), and
\ref{fig:proof_maxtriangles_567}(\subref{fig:proof_maxtriangles_7}), respectively. The digons created
by the flip are marked ``D''.}
  \label{fig:proof_maxtriangles_911}
\end{figure}
%%%%%%%%%%%%%%%%%%%%%%%%%%%%%%%%%%%%%%%%%%%%%%%%%%%%%%%%%%%%%%%%%%%%%%%%%%%%

To show that $|X'|$ is small, we are going to trade crossings of $X'$
with digons and then refer to a result of Agarwal et
al.~\cite{anppss-lenses-04}. They have shown that the number of digons
in intersecting arrangements of pseudocircles is at most linear in $n$.

To convert crossings of $X'$ into digons we use \emph{triangle flips}.
Each of the configurations shown in
Figure~\ref{fig:proof_maxtriangles_567}
has a gray triangle. By flipping these triangles we obtain the 
configurations shown in Figure~\ref{fig:proof_maxtriangles_911}.
These so-obtained configurations have at least as many new digons as the
original configurations contain crossings in~$X'$. 
It may be that the flip creates new triangles and even
new vertices which are incident to 3 triangles. However,
the flips never remove digons. 

Therefore, thanks to the
result from~\cite{anppss-lenses-04} we can make 
no more than $O(n)$ flips before all the crossings are incident to
at most 2 triangles.
This finishes the proof of Theorem~\ref{thm:upper-max}.
\end{proof}

In the proof of Theorem~\ref{thm:upper-max}, we have used flips to trade
segments incident to two triangles against digons. It can be shown that at
most one component of the graph induced by $X'$ is a $K_3$. The
proof of this fact is omitted here since it does not improve the 
bound given in the theorem. Having used a bound on the number of digons 
we recall that Gr\"unbaum conjectures that $p_2 \le 2n-2$ holds
for intersecting arrangements.

Since intersecting arrangements have $2\binom{n}{2}+2$ faces, we
can also rewrite the statement of Theorem~\ref{thm:upper-max}: at most
$\frac{2}{3}+O(\frac{1}{n})$ of all cells of an intersecting arrangement are
triangles.  For $n=7$ there exist arrangements with
$29=\frac{4}{3}\binom{7}{2}+1$ triangles. It would be interesting to
know what the precise maximum value of $p_3$ for $n$ large.
   
For non-intersecting arrangements the arguments from the proof of
Theorem~\ref{thm:upper-max} do not work.
Figure~\ref{fig:interesting_connected_arrangements}(\subref{fig:connected_arrangement_all_triangles})
shows an arrangement where all but two cells are triangles.  However, if each
pseudocircle is required to intersect at least 3 other pseudocircles, then we
can proceed similar and show that the triangle-cell-ratio is at
most~$5/6+O(1/n)$.  In fact,
Figure~\ref{fig:interesting_connected_arrangements}(\subref{fig:connected_arrangement_56_triangles})
shows a construction with triangle-cell-ratio $5/6+O(1/\sqrt{n})$.
  
\begin{theorem}\label{thm:56_upper_bound}
Let $\AA$ be an arrangement of $n$ pseudocircles
where every pseudocircle intersects at least three other pseudocircles.
Then the triangle-cell-ratio is at most~$5/6+O(1/n)$.
\end{theorem}

\begin{proof}
We proceed as in the proof of Theorem~\ref{thm:upper-max}.
In fact, as the ``intersecting'' property was only used to bound the number of digons,
Claims~\ref{thm:upper-max_claim:A}--\ref{thm:upper-max_claim:C} hold also in this less restrictive setting.

From Claims~\ref{thm:upper-max_claim:A}--\ref{thm:upper-max_claim:C} 
we have learned that every vertex from $X'$ has at least two neighbors from $X \setminus X'$.
The following claim will help us to show $|X'| \le |X \setminus X'|$.

\addtocounter{claim}{3}
\begin{claim}
\label{thm:upper-max_claim:D}
Every vertex from $X \setminus X'$ has at most two neighbors from~$X'$.
\end{claim}

Suppose for a contradiction that a vertex $v \in X \setminus X'$ has (at
least) three neighbors $x,y,z$ from~$X'$.  Since $x,y,z$ each have three
incident triangular faces and since $v \not \in X'$,  two of the
neighboring faces of~$v$ are triangles.  In particular, those two triangular
faces are not adjacent as otherwise $x,y,z$ would lie in the same component of
$G[X']$ and have the same neighbor $v$ -- which is impossible.

Without loss of generality, we assume that $xy$ is an edge and that $z$ forms
an edge with the fourth neighbor of $v$, which we denote by~$w$.  Since $x$
is incident to a non-triangular face (which is also incident to $v$), the edge
$xy$ bounds another triangle.  The same argument shows that $zw$ bounds
another triangle, and therefore, the two pseudocircles passing through $v$
intersect three times -- a contradiction; see
Figure~\ref{fig:proof_maxtriangles_claim_IV}.  This finishes the proof of
Claim~\ref{thm:upper-max_claim:D}.  \qedclaim

%%%%%%%%%%%%%%%%%%%%%%%%%%%%%%%%%%%%%%%%%%%%%%%%%%%%%%%%%%%%%%%%%%%%%%%%%%%%
\begin{figure}[htb]
  \centering
  %\begin{subfigure}[t]{.25\textwidth}
   % \centering
    \includegraphics[page=12]{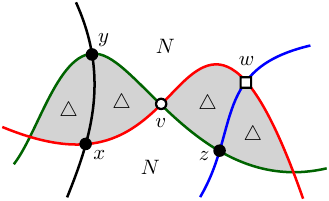}
    %\caption{}
    %\label{fig:proof_maxtriangles_12}
  %\end{subfigure}
\caption{An illustration of the proof of Claim~\ref{thm:upper-max_claim:D}.}
  \label{fig:proof_maxtriangles_claim_IV}
\end{figure}
%%%%%%%%%%%%%%%%%%%%%%%%%%%%%%%%%%%%%%%%%%%%%%%%%%%%%%%%%%%%%%%%%%%%%%%%%%%%

We can now discharge $1/2$ from every vertex of $X'$ to its neighbors from $X
\setminus X'$ and, by Claim~\ref{thm:upper-max_claim:D}, count at most~1 at
each of the vertices from $X \setminus X'$.  Therefore, $|X'| \le |X \setminus
X'|$ holds. By counting the face-vertex-incidences we get
$$
p_3 \le \frac{3|X'|+2|X \setminus X'|}{3}\le  \frac{5|X|}{6} ,
$$
and since the number of faces equals $|X|+2$, 
this completes the proof of Theorem~\ref{thm:56_upper_bound}.
\end{proof}

\subsection{Constructions using Arrangements of Pseudolines}
\label{sec:great-pseudocircles}

Great-circles on the sphere are a well known model for projective arrangements of lines.
Antipodal pairs of points on the sphere
correspond to points of the projective plane. Hence, the great-circle
arrangement corresponding to a projective
arrangement $\AA$ of lines has twice as many vertices, edges, and faces
of every type as $\AA$. The same idea can be applied to 
projective arrangements of pseudolines. If $\AA$ is a projective
arrangement of pseudolines, take a drawing of $\AA$ in the unit disk $D$
such that every line $\ell$ of~$\AA$ connects two antipodal points of~$D$. 
Project $D$ to the upper hemisphere of a sphere $S$, such that the
boundary of $D$ becomes the equator of $S$. Use a projection through
the center of $S$ to copy the drawing from the  upper hemisphere
to the lower hemisphere of $S$. By construction the two copies of each
pseudoline from $\AA$ join together to form a pseudocircle.
The collection of these pseudocircles yields an intersecting arrangement of
pseudocircles on the sphere with twice as many vertices, edges, and faces
of every type as $\AA$. Arrangements of
pseudocircles obtained by this construction have a special property: 
\begin{itemize}
\item 
If three pseudocircles $C$, $C'$, and $C''$ have no common crossing,
then $C''$ separates the two crossings of 
$C$ and $C'$.
\end{itemize}

\noindent
Gr\"unbaum~\cite{Gr72} calls arrangements with this property `symmetric'.  In
the context of oriented matroids the property is part of the definition of
arrangements of pseudocircles~\cite{blwsz-om-93}.
In~\cite{FelsnerScheucher2019} we call arrangements with this
property ``arrangement of great-pseudocircles'' as they generalize the
properties of arrangements of great-circles.

Arrangements of pseudolines which maximize the number of triangles have been
studied intensively. Blanc~\cite{Blanc2011} gives tight upper bounds for the
maximum both in the Euclidean and in the projective case and constructs
arrangements of pseudolines with $\frac{2}{3}\binom{n}{2} - O(n)$ triangles
for every~$n$.  In particular, for $n \equiv 0,4 \pmod 6$ projective
arrangements of straight lines with $\frac{2}{3}\binom{n}{2}$ triangles are
known; see also~\cite{fg-pa-16}.  This directly translates to the existence of
(1)~intersecting arrangements of pseudocircles with $\frac{4}{3}\binom{n}{2} -
O(n)$ triangles for every~$n$ and (2)~intersecting arrangements of circles
with $\frac{4}{3}\binom{n}{2}$ triangles for $n \equiv 0,4 \pmod 6$.  The
`doubling method' that has been used for constructions of arrangements of
pseudolines with many triangles, see~\cite{Blanc2011}, can also be applied for
pseudocircles. In fact, in the case of pseudocircles there is more flexibility
for applying the method. Therefore, it is conceivable that
$\lfloor\frac{4}{3}\binom{n}{2}\rfloor$ triangles can be achieved for all~$n$.

%%%%%%%%%%%%%%%%%%%%%%%%%%%%%%%%%%%%%%%%%%%%%%%%%%%%%%%%%%%%%%%%%%%%%%%%%%%%
\section{Visualization}
\label{sec:tutte-draw}

Most of the figures in this paper have been generated automatically.  The
programs are written in the mathematical software
SageMath~\cite{sagemath_website}, they are available on demand.  We encode an
intersecting arrangement of pseudocircles by its dual graph.  Each face in the
arrangement is represented by a vertex and two vertices share an edge if and
only if the two corresponding faces share a common pseudosegment.  Note that,
in the dual graph of every intersecting arrangement, the only 2-separators are
the two neighbored vertices of a vertex corresponding to a digon.  By
replacing such digon-vertices by edges, we obtain a 3-connected graph which has
the ``same'' embeddings as the original graph.  Since 3-connected planar
graphs have a unique embedding (up to isomorphism), the same is true for the
original dual graph. 

To visualize an intersecting arrangement of pseudocircles,
we draw the primal (multi)graph using straight-line segments, 
in which vertices represent crossings of pseudocircles 
and edges connect two vertices if they are connected by a pseudocircle segment. 
Note that in the presence of digons we obtain double-edges.

In our drawings, 
pseudocircles  are colored by distinct colors,
and triangles (except the outer face) are filled gray.
In straight-line drawings, edges corresponding to digons 
are drawn dashed in the two respective colors alternatingly,
while in the curved drawings digons are represented by a
point where the two respective pseudocircles touch.

\subsection{Iterated Tutte Embeddings}

To generate nice aesthetic drawings automatically, we iteratively use
weighted Tutte embeddings.  We fix a non-digon cell as the outer cell
and arrange the vertices of the outer cell as the corners of a regular
polygon.  Starting with edge-weights all equal to 1, we obtain an
ordinary plane Tutte embedding.

For iteration $j$,
we set the weights (force of attraction) of an edge $e=\{u,v\}$ 
proportional to $p(A(f_1)) + p(A(f_2)) + q(\|u-v\|/j)$
where $f_1,f_2$ are the faces incident to $e$, $A(.)$ is the area
function, $\| \cdot \|$ is the Euclidean norm, and $p,q$ are suitable monotonically 
increasing functions from $\mathbb{R}^+$ to $\mathbb{R}^+$ (we use $p(x) = x^4$ and $q(x)=x^2/10$).

Intuitively, if the area of a face becomes too large,
the weights of its incident edges are increased and will rather be
shorter so that the area of the face will also get smaller in the next
iteration. It turned out that in some cases the areas of the faces 
became well balanced but some edges were very short and others
long. Therefore we added the dependence on the edge length which is
strong at the beginning and decreases with the iterations.
The particular choice of the functions was the result of
interactive tuning. The iteration is terminated when the change of 
the weights becomes small enough.

\subsection{Visualization using Curves}

On the basis of the straight-line embedding obtained with the Tutte
iteration we use splines to smoothen the curves. The details are 
as follows. First we take a 2-subdivision of the graph,
where all subdivision-vertices adjacent to a given vertex $v$
are placed at the same distance $d(v)$ from $v$. 
We choose $d(v)$ so that it is at most 1/3 of the length of an edge
incident to $v$. We then use B-splines to visualize the curves.
Even though one can draw B\'ezier curves directly with Sage,
we mostly generated ipe files (xml-format, cf.\ \cite{IPE}) so that we 
can further process the arrangements.
Figures~\ref{fig:n9_nonr_pappus}(\subref{fig:n9_nonr_pappus_pca_tutte_straightline}) 
and \ref{fig:n9_nonr_pappus}(\subref{fig:n9_nonr_pappus_pca_tutte})
show the straight-line and curved drawing of an arrangement of pseudocircles, respectively.

%%%%%%%%%%%%%%%%%%%%%%%%%%%%%%%%%%%%%%%%%%%%%%%%%%%%%%%%%%%%%%%%%%%%%%%%%%%%
\begin{figure}[htb]
  \centering
  \begin{subfigure}[t]{.32\textwidth}
    \centering
    \includegraphics[width=\textwidth]{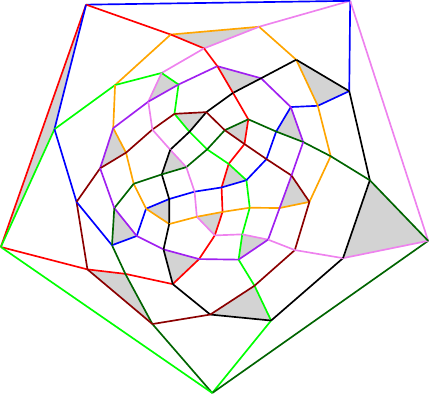}
    \caption{}
    \label{fig:n9_nonr_pappus_pca_tutte_straightline}
  \end{subfigure}
  \hfill
  \begin{subfigure}[t]{.32\textwidth}
    \centering
    \includegraphics[width=\textwidth]{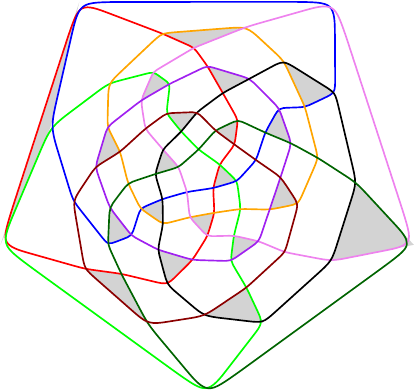}
    \caption{}
    \label{fig:n9_nonr_pappus_pca_tutte}
  \end{subfigure}
  \hfill
  \begin{subfigure}[t]{.3\textwidth}
    \centering
    \includegraphics[width=\textwidth]{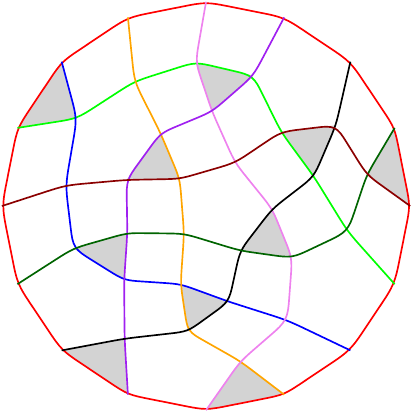}
    \caption{}
    \label{fig:n9_nonr_pappus_pla_tutte}
  \end{subfigure}
  \caption{(\subref{fig:n9_nonr_pappus_pca_tutte_straightline})~Straight-line and (\subref{fig:n9_nonr_pappus_pca_tutte})~curved drawings 
  of the arrangement of great-pseudocircles, which consists of two copies of
  (\subref{fig:n9_nonr_pappus_pla_tutte})~the non-Pappus arrangement of pseudolines.}
  \label{fig:n9_nonr_pappus}
\end{figure}
%%%%%%%%%%%%%%%%%%%%%%%%%%%%%%%%%%%%%%%%%%%%%%%%%%%%%%%%%%%%%%%%%%%%%%%%%%%%

\subsection{Visualization of Arrangements of Pseudolines}

We also adapted the code to visualize arrangements of
pseudolines nicely.  One of the lines is considered as the ``line at infinity''
which is then drawn as a regular polygon.  
Figure~\ref{fig:n9_nonr_pappus}(\subref{fig:n9_nonr_pappus_pla_tutte}) gives an illustration.

\subsection{A more general Representation}

As suggested in~\cite{FelsnerScheucher2019},
intersecting arrangements with digons and
non-intersecting arrangements of pseudocircles 
may be visualized by their primal-dual graph;
see for example Figure~\ref{fig:two_repesentations_of_AAsixA}.
Even though the primal-dual graph is a simple graph and has a unique embedding,
we decided to stick to the above described visualizations 
because primal-dual graphs have about 4 times as many primitives as dual graphs
and therefore are somewhat harder to read (for humans).
In particular, $k$-cells in the arrangement are visualized as polygons of size $2k$
and therefore that representation is not that suitable for an article on cells.
As an example consider the rightmost triangle bounded 
by the green, the orange, and the black pseudocircle in
Figure~\ref{fig:two_repesentations_of_AAsixA}(\subref{fig:three_repesentations_of_AAsixA_2})
which actually looks like a quadrangle.

%%%%%%%%%%%%%%%%%%%%%%%%%%%%%%%%%%%%%%%%%%%%%%%%%%%%%%%%%%%%%%%%%%%%%%%%%%%%
\begin{figure}[htb]
  \centering
  
  \hbox{}
  \hfill
  \begin{subfigure}[t]{.3\textwidth}
    \centering
    \includegraphics[width=\textwidth]{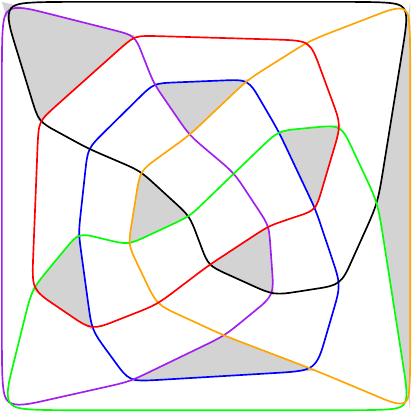}
    \caption{}
    \label{fig:three_repesentations_of_AAsixA_1}
  \end{subfigure}
  \hfill
  \begin{subfigure}[t]{.32\textwidth}
    \centering
    \includegraphics[page=3,width=\textwidth]{n6_nonr_number1_8triangles}
    \caption{}
    \label{fig:three_repesentations_of_AAsixA_2}
  \end{subfigure}
  \hfill
  \hbox{}

  \caption{Two drawings of $\AAsixA$: 
  (\subref{fig:three_repesentations_of_AAsixA_1})~curved primal graph. 
  (\subref{fig:three_repesentations_of_AAsixA_2})~curved primal-dual graph.}
  \label{fig:two_repesentations_of_AAsixA}
\end{figure}
%%%%%%%%%%%%%%%%%%%%%%%%%%%%%%%%%%%%%%%%%%%%%%%%%%%%%%%%%%%%%%%%%%%%%%%%%%%%

%%%%%%%%%%%%%%%%%%%%%%%%%%%%%%%%%%%%%%%%%%%%%%%%%%%%%%%%%%%%%%%%%%%
%%\small
\advance\bibitemsep-3pt
\let\sc=\scshape
\bibliographystyle{my-siam}
\bibliography{bibliography}

\end{document}